\documentclass[final]{IEEEtran}
\usepackage[dvips,final]{graphicx}
\usepackage{latexsym}
\usepackage{amssymb}
\usepackage{amsfonts}
\usepackage{amsmath,epsfig,algorithmic,color}
\usepackage{subfigure}
\usepackage{tikz}
\usetikzlibrary{arrows,shapes,plotmarks,backgrounds,automata}
\usepackage{algorithm}
\newtheorem{lemma}{{Lemma}}
\psfull

\title{Intermediate Performance Analysis of Growth Codes
\thanks{N. Thomos is with the Signal Processing Laboratory 4 (LTS4), Ecole Polytechnique F\'ed\'erale de Lausanne (EPFL), Lausanne, Switzerland, and the Communication and Distributed Systems laboratory (CDS), University of Bern, Bern, Switzerland (e-mail: nikolaos.thomos@epfl.ch).}
\thanks{R. Pulikkoonattu is with Broadcom Corp, USA (e-mail: rethna@broadcom.com).}
\thanks{P. Frossard is with the Signal Processing Laboratory 4 (LTS4), Ecole Polytechnique F\'ed\'erale de Lausanne (EPFL), Lausanne, Switzerland (e-mail: pascal.frossard@epfl.ch).}
\thanks{This work has been supported by the Swiss National Science Foundation, under grants PZ00P2-121906 and PZ00P2-137275.}}

\author{Nikolaos Thomos, \textit{Member, IEEE}, Rethnakaran Pulikkoonattu, \textit{Member, IEEE}, and Pascal Frossard, \textit{Senior Member, IEEE}.}

\begin{document}

\maketitle

\begin{abstract}
Growth codes are a subclass of Rateless codes that have found interesting applications in data dissemination problems. Compared to other Rateless and conventional channel codes, Growth codes show improved intermediate performance which is particularly useful in applications where performance increases with the number of decoded data units. In this paper, we provide a generic analytical framework for studying the asymptotic performance of Growth codes in different settings. Our analysis based on Wormald method applies to any class of Rateless codes that does not include a precoding step. We evaluate the decoding probability model for short codeblocks and validate our findings by experiments. We then exploit the decoding probability model in an illustrative application of Growth codes to error resilient video transmission. The video transmission problem is cast as a joint source and channel rate allocation problem that is shown to be convex with respect to the channel rate. This application permits to highlight the main advantage of Growth codes that is improved performance (hence distortion in video) in the intermediate loss region.
\end{abstract}


\section{Introduction}
\label{sec:intro}

Rateless codes \cite{DigitalFountain98} have been proposed as an efficient method to design decentralized data dissemination systems without employing expensive routing protocols as they do not require any source coordination. This becomes feasible due to the rateless property that allows on-the-fly generation of arbitrary numbers of packets. The most successful Rateless codes are certainly the Raptor codes \cite{Shokrollahi06} that are based on LT codes \cite{Luby02}. Raptor codes perform close to maximum distance separable (MDS) codes and only require slightly more symbols than the number of source symbols for successful recovery. These appealing properties, combined with linear encoding and decoding times, have conducted the Raptor codes to be adopted by recent communication protocols like $3$GPP \cite{3GPP} and DVB-H \cite{DVB-H}.

In practice, Rateless codes need a feedback channel for transmitting a termination message from sinks to sources in order to notify that the data decoding is successful. Unfortunately, in many scenarios a feedback channel is not available. Thus, the transmission rate should be predetermined by taking into account the network conditions. However, the channel conditions vary dynamically, may deteriorate fast and become rapidly different from the predicted ones. This leads to inaccurate estimation of the transmitted rate, and thus to inefficient exploitation of the network resources with possibly decoding failures. If the decoder receives an insufficient number of packets for forming a full rank decoding matrix, common Rateless codes have unfortunately poor recovery properties; they are characterized by on-off performance, which means that the transmitted data can be either fully recovered or not at all. Partial decoding is possible, but in general it is limited. Growth codes \cite{Growthcodes} have been proposed to enhance the intermediate performance of Rateless codes. They permit recovery of more data than conventional codes when the available data is insufficient for complete message recovery. The improved intermediate performance of Growth codes has attracted the attention of the data communication community, especially for the transmission of unequal importance data streams as proposed in \cite{UnequalGrowthcodes}. 

In this paper, we propose a new and complete decoding performance analysis for Growth codes based on the Wormald method \cite{Wormald95} that is typically used for the analysis of graph stochastic processes. Contrarily to previous studies, the proposed method for characterizing the decoding performance of Growth codes is generic and valid for various codeblocks\footnote{In this paper, the term codeblock is used to denote a group of packets that are encoded together by a channel code.} size and accurately predicts the decoding probability for both intermediate and high reception rates (when more symbols than the number of source symbols have been received). It therefore provides a more generic performance analysis than asymptotic performance bounds with arbitrary long blocks such as those derived in \cite{Sanghavi07}. An accurate performance analysis of small and medium blocksize is very important in the design of data communication systems and our novel performance analysis shows that the decoding performance of Growth codes can be approximated by an exponential model. We apply our analysis in an illustrative example in video communications. We combine our model with the video distortion model proposed in \cite{paper:frossard:2001} and cast a joint source and channel rate allocation problem for optimized error resilient transmission. This problem is shown to be convex and the optimal source and channel rate can be determined by methods such as dynamic programming and Lagrangian optimization. Our design comes close to the concepts of \cite{UnequalRateless,EWFTMM09} where unequal error protection Rateless codes have been designed specifically for streaming. We however here concentrate on the intermediate of the channel codes, which is one of the main advantages of Growth codes. We show that these codes provide an interesting solution for error resilient video streaming with graceful degradation of quality when channel conditions deteriorate. 

The paper is organized as follows. In Section~\ref{sec:Rateless_codes}, we present the general principles of Rateless codes focusing especially on Growth codes. We analyze Growth codes recovery performance using the Wormald method, discuss the expected Growth codes behavior and provide a model for the symbol recovery probability in Section~\ref{sec:wormald}. In Section~\ref{sec:distanalysis}, we cast an optimization problem for determining the optimal source and channel rate allocation in an illustrative example of video transmission. Finally, conclusions are drawn in Section~\ref{sec:conclusions}.

\section{Rateless codes}
\label{sec:Rateless_codes}

\subsection{General structure}

Data dissemination systems usually employ Reed Solomon (RS) codes \cite{Lin82} to cope with packet erasures on communication channels. RS codes can recover the transmitted data only when the received set of packets is equal to the number of source symbols $k$. However, the application of RS codes is limited by their decoding complexity that grows quadratically with the block size. Rateless codes solve this problem in applications with complexity or delay constraints as they offer linear encoding times. Linear decoding times can be further achieved when the degree distribution is carefully designed and decoding is performed by means of belief propagation. These codes however incur a small performance penalty $\epsilon$ and require slightly more packets than the number of source packets $k$ for decoding, (\textit{i.e.}, the minimal number of packets for decoding is $(1+\epsilon) \cdot k$). The decoding probability increases with the number of received packets since the probability to form a full rank system of packets at decoder increases with the number of packets. Rateless codes generally do not have a pre-determined code rate and they generate as many symbols as needed by each client for successful decoding. This permits better exploitation of the available network resources as systems employing Rateless codes can easily cope with network dynamics; inaccurate estimation of the channel conditions do not directly lead to over-protecting or under-protecting the source data as in rate-optimized codes.

In more details, the Rateless encoded symbols are generated by combining (XOR-ing) source packets selected uniformly at random. The number of combined packets (the degree of the corresponding codeword) is determined by the degree distribution function $\Omega(x)=\sum^{d_{\max}}_{i=1} \Omega_i \cdot x^i$, where $\Omega_i$ denotes the probability of generating a symbol with degree $i$ (a symbol that is the result of XOR-ing $i$ data symbols). The parameter $d_{\max}$ is the maximum allowable symbol's degree which cannot exceed the number of source symbols. If $x_i$ stands for the $i$th source packet and $y_j$ is the $j$th transmitted symbol respectively, we can write $$y_j=\sum^{d_j}_{l=1}\bigoplus\ x_l,$$ where $d_j$ is the degree of the $j$th symbol and $\sum\bigoplus$ is the bitwise XOR operation. Rateless codes have an implicit structure as the decision about the number of packets to be combined for the generation of the $y_j$ is made by sampling $\Omega(x)$; the identity of the combined packets is selected uniformly at random. Due to the implicit structure of Rateless codes a small header called ``ESI '' is appended to each packet. It conveys information to the decoder about the packets that have been combined for the generation of encoded packets. The ESI is usually the seed of a pseudorandom generator used for generating the encoded packet. Please note that the terms ``packet'' and ``symbol'' are used interchangeably throughout the analysis.

Next, we focus on Growth codes that are a particular class of Rateless codes with advanced intermediate decoding performance and sustained performance at high rates. Such characteristics are particularly interesting for transmission of time-constrained data that can benefit from partial decoding. 

\subsection{Growth codes}
\label{sec:growth}

Popular Rateless codes such as LT codes and Raptor codes have small overheads and permit decoding from a set of symbols than is slightly larger than the set of source packets. These codes have however poor intermediate performance when the number of symbols is not sufficient for perfect decoding. On the other hand, Growth codes \cite{Growthcodes} offer better intermediate performance as partial recovery is possible, but are characterized by a larger overhead $\epsilon$ than LT or Raptor codes.

We give now more details about the design of Growth codes. The encoding procedure of Growth codes is similar to that of LT codes. Symbols are generated according to a degree distribution function $\Omega(x)$. The design of $\Omega(x)$ is quite intuitive, \textit{i.e.}, when few packets have arrived to a client, it is better to receive packets of degree one that  permit immediate decoding. Instantaneous decodability of packets of degree one however decreases with the number of received packets like in the coupons collector problem \cite{RandomizedAlgorithms95}. Indeed, the probability that a packet is a duplicate increases with the number of received packets with degree one \cite{Growthcodes}.

Based on this intuition, the degree of the encoded symbols progressively increases until sinks are able to recover the transmitted data content. The Growth codes degree distribution is thus given as

\begin{equation}
\Omega(k): \Omega_i^{*} = \max
\left(0,\min\left(\frac{K_i-K_{i-1}}{k},\frac{k-K_{i-1}}{k}\right)\right)
\label{eqn:growth}
\end{equation}
where the parameter $K_i$ is computed by the following recursive relation \cite{Growthcodes} for the $k$th transmitted symbol:

\begin{equation}
K_j=K_{j-1}+\sum_{i=R_{j-1}}^{R_j-1}{\frac{\binom{k}{j}}{\binom{i}{j-1}
(k-i)}}
\end{equation}
with $R_j=\frac{jk-1}{j+1}$ and $K_1=\displaystyle \sum_{i=0}^{R_1-1}{\frac{k}{k-i}}$. In order to recover $R_j$ symbols, the receiver has to get $K_j$ symbols in expectation. The Growth codes decoding probability as a function of the number of received packets $k$ is determined by

\begin{equation}
P_{r,d}(k) = \begin{cases} \frac{(k-r+1)}{k}, & d=1 \\ \frac{(k-r+1)\binom{r-1}{d-1}}{\binom{k}{d}}, & d=2,\ldots,r \\ 0, & d> r \end{cases}
\label{eqn:prd}
\end{equation}
It corresponds to the probability of decoding a symbol of degree $d$ when $r$ symbols have been already decoded. The Growth codes, encoding procedure, can be summarized as follows

\floatname{algorithm}{Procedure}
\begin{algorithm}
\caption{Growth Codes encoding}
\label{algo:Growth_encode}
\begin{algorithmic}[1]
    \STATE Choose randomly the degree $d$ of the LT encoded symbol by sampling $\Omega(d)$.
    \STATE Choose uniformly $d$ distinct symbols.
    \STATE Combine the symbols by XOR operations
\end{algorithmic}
\end{algorithm}

On the client side, two algorithms can be used for Growth codes decoding namely the Decoder-S and the Decoder-D \cite{Growthcodes}. The Decoder-D ignores the packets that are not immediately useful, \textit{i.e.}, the packets whose distance to the recovered packets is larger than one, while Decoder-S stores in a list all received packets even if they are not immediately useful and checks the list of packets as the decoding process progresses. Although Decoder-D has slightly inferior decoding performance compared to that of Decoder-S, it is preferable for low-cost receivers performing on-the-fly decoding. In our analysis, we restrict our attention only to the Decoder-D that is summarized in Procedure \ref{algo:Growth_decode}.

\floatname{algorithm}{Procedure}
\begin{algorithm}[h]
\caption{Growth Codes decoding: Decoder-D algorithm}
\label{algo:Growth_decode}
\begin{algorithmic}[1]
    \IF{a symbol $y$ with $d=1$ arrives at client}
    \STATE Insert the symbol $y$ in $\mathcal{D}$.
    \ENDIF
    \WHILE{$\mathcal{D} \neq \O$}
        \STATE Select a symbol $x$ from $\mathcal{D}$.
        \IF{$x \notin \mathcal{L}$ and $\text{dist}(x,\mathcal{L})=1$}            
            \STATE XOR $x$ with the symbols $\in \mathcal{L}$ that used to generate the symbol $x$.
            \STATE Insert the recovered symbol in $\mathcal{L}$.
            \IF{in the meanwhile a new instantaneous decodable symbol $y$ arrives at client}
            \STATE Insert the symbol $y$ in $\mathcal{D}$.
            \ENDIF
        \ENDIF
    \ENDWHILE
\end{algorithmic}
\end{algorithm}

In the decoding algorithm, $\mathcal{D}$ is the ripple of the Growth codes, $\mathcal{L}$ is a list that contains the decoded symbols and $\text{dist}(x,\mathcal{L})=1$ means that the symbol is instantaneously decodable, \textit{i.e.}, an original source symbol can be recovered. Upon receiving a packet, the client examines whether the packet is instantaneously decodable and is not already in $\mathcal{L}$. When it is true, the packet is decoded and then inserted in $\mathcal{L}$, otherwise the packet is deleted. This procedure continues as long as the source is decoded or $\mathcal{D}$ is empty.

\section{Analysis of Growth Codes Performance}
\label{sec:wormald}

\subsection{Wormald method}
 
The Wormald method \cite{ModernCodingTheory} is an analytical tool that is used to study the expected behavior of stochastic processes. It is based on the idea that a system stays close to the expected behavior \cite{Wormald95} after a series of random steps, with very high probability. Such a behavior can further be determined by a set of differential equations.

Let us consider a graph random process $G(t)$. The process starts with graph $G(0)$ from which edges are repeatedly removed according to a probabilistic rule that is known a priori. This removal procedure results in a probabilistic set of sequences $G(0),G(1),\ldots, G(t)$, where $G(t)$ denotes the $t$-th graph in that process. When no edge can be removed from $G(t)$, the process becomes stationary. Thus, we have $G(t+1)=G(t)$ and this is the final graph of this stochastic process. Usually this happens for large $t$ ($t \to \infty$). For many systems the above procedure however terminates after a finite number of steps, \textit{i.e.}, at time $t=T_d$; this forms a family of graphs $G(0),G(1),\ldots,G(T_d)$. The Wormald method studies the behavior of such a sequence of processes.

Since Decoder-D algorithm for Growth codes decoding is described by a stochastic process (\textit{i.e.}, Markov process), the Wormald method can be used for the analysis of its performance. In particular, the residual Tanner graph $G(t)$ during decoding is characterized by a set of pairs $(V_i(t),C_i(t))$, where $V_i(t)$ and $C_i(t)$ denote the total number of edges connected with variable nodes and respectively check nodes of degree $i$ at time $t$. In the Growth Codes analysis, the graph actually contains $k$ variable nodes where $k$ is equal to the number of source symbols.

Next, we discuss the appropriateness of using Wormald theorem to characterize the decoding performance of Growth codes. We start by giving the conditions that Growth codes decoding should fulfill in order to use the Wormald method. Then, we use this method to describe the evolution of the decoding process.

\subsection{Wormald theorem}

In this section, we describe the Wormald theorem and present the conditions that have to be satisfied to apply the Wormald method \cite{Lipschitz}. It should be noted that the Wormald theorem can also be used for analyzing other Rateless codes, as long as their encoding does not include any precoding step. 

Let us consider that $G(t)$ has  a state space $\{0,\ldots,\theta\}^{d}, d\in \mathbb{N}$ and a probability space $\mathcal{S}$. Consider a sequence $\left\{G^{m}(t)\right\}_{m>1}$ of a Markov random process
where $G_{i}^{m}(t)$ is the $i$-th component of $G^{m}(t)$. Denote a subset $\Gamma \subset \mathbb{R}^{d+1}$ containing the vectors $\left[0,g_{1},\ldots,g_{d}\right]$ such that,
\begin{equation}
\mathbb{P}\left( \frac{G_{i}^{(m)} \left(t=0\right)}{m} =g_{i}\right) > 0, \forall i \in [1,d], m > 1
\end{equation}

Let $f_{i}$ be functions from $\mathbb{R}^{d+1}$ to $\mathbb{R}^{d}$, satisfying the following conditions:

\begin{enumerate}
\item For $t <m$, there exists a constant $c_{i}^{m}$ such that,
\begin{equation}
\left \lvert G_{i}^{(m)}(t+1) - G_{i}^{(m)}(t)\right \rvert \le c_{i}^{m}, 1\le i \le d
\end{equation}
\item For $t <m$ and $\forall i\in[1,d[$,
\begin{eqnarray}
&{}& \mathbb{E} \left[G_{i}^{(m)}(t+1) - G_{i}^{(m)}(t) \rvert G^{(m)}(t) \right] \nonumber \\
&{}& \triangleq f_{i}\left(\frac{t}{m},\frac{G_{1}^{(m)}(t) }{m}\ldots,\frac{G_{d}^{(m)}(t) }{m}\right)
\end{eqnarray}
\item $f_{i}, \forall i \le d$ is Lipschitz continuous function on the intersection of $\Gamma$ with the half space
$\left\{\left(t,g_{1},\ldots,g_{d}\right):t\ge 0\right\}$, \textit{i.e.}, if $x,y \in \mathbb{R}^{d+1}$ belong to this
intersection, then there exists a Lipschitz constant $\zeta$ such that,
\begin{equation}
\left \lvert f_{i}(x)-f_{i}(y)\right \rvert \le \zeta \sum_{j=1}^{d+1}{\left \lvert x_{j}-y_{j}\right \rvert}
\end{equation}
\end{enumerate}

The above conditions are respectively the boundedness, the trend and the Lipschitz conditions. The boundedness condition implies that the function is bounded. The trend condition imposes that we can describe the process with a time series without any knowledge of the serial correlation between $G^m(t)$. Finally, the Lipschitz condition means that the functions $f_i$ are not steep and bounded by $\zeta$. Under these conditions, the following holds true:
\begin{enumerate}
\item For the vector $\left[t,g_{1},g_{2},\ldots,g_{d}\right] \in \Gamma$, the system of differential equations
\begin{equation}
\frac{\partial g_{i}}{\partial \tau } = f_{i}\left(\tau,g_{1},\ldots,g_{d}\right), 0\le i \le d
\label{eqn:lip}
\end{equation}
has a unique solution for $g_{i}(\tau):\mathbb{R} \to \mathbb{R}$  in $\Gamma$ with the initial condition
$g_{i}(0)=x_{i}, 1\le i \le d.$
\item There exists a strictly positive constant $\delta$ such that
\begin{equation}\label{eqn:monster}
\mathbb{P} \left(\left \lvert \frac{G_{i}^{(m)}(t)}{m}- g_{i} \left( \frac{t}{m} \right)\right \rvert \ge \delta m^{-\frac{1}{6}}\right)  < \frac{d m^{\frac{2}{3}}} {e^{\frac{m^{\frac{1}{3}}}{2}}}
\end{equation}
for $0 \le t \le m \tau_{\max}$ and for each $i$. The term $g_{i}$ is the unique solution obtained by solving Eq. (\ref{eqn:lip}) with the initial conditions $g_{i}(0)=\mathbb{E}\left[\frac{G_{i}^{(m)}(t=0)}{m}\right]$
and $\tau_{\max}=\tau_{m}$ is the supremum of those $\tau$ to which the solution can be extended, under some boundedness criteria \cite{Wormald95,ModernCodingTheory}.
\end{enumerate}

In other words, Eq. (\ref{eqn:monster}) states that, when $m$ is big enough, each realization of the process $G_{i}^{(m)}(t)$ is close to the (unique) solution of Eq. (\ref{eqn:lip}), with high probability. This permits us to use a set of differential equations to describe the decoding process. A formal proof of the applicability of the Wormald theorem on the Decoder-D algorithm can be further found in \cite{ModernCodingTheory}.

\subsection{Expected Growth Codes behavior}
\label{sec:symdecprob}

We have seen above that a set of differential equations can describe the expected decoding performance as the decoding procedure evolves. In particular, the Decoder-D algorithm proceeds as long as it can find check nodes of degree one to process, otherwise the decoding halts. At each decoding step a check node with degree one is randomly chosen to be eliminated. The outgoing edge of the selected node is connected with one of the $\sum_{j}V_{j}(t)$ variable nodes, while the total number of edges $E(t)$ in the (residual) graph $G_t$ is equal to $\sum_{j} j V_{j}(t)$. Thus, the probability that an edge is connected to a variable node of degree $i$ is $\frac{i V_{i}(t)}{\sum_{j} j V_{j}(t)}$.

The expected decrease in the number of variable nodes with  degree $i$ at time $t$ is then expressed as,
\begin{equation}
\mathbb{E}\left[V_{i}(t+1)-V_{i}(t)\lvert V(t),E(t) \right] =
-\frac{i V_{i}(t)}{\sum_{j} j V_{j}(t)}
\end{equation}

Similarly to the variable nodes, we can determine the expected decrement in check nodes degree. When an edge of a check node of degree one is removed, the number of check nodes of degree one is reduced by one. If the removed edge is connected to a check node of degree $i+1$, the residual degree changes from $i+1$ to $i$ with probability $\frac{(i+1) C_{i+1}(t)}{\sum_{j}jC_{j}(t)}$. Therefore, in expectation, $-1+\frac{\sum_{j}{j^{2} V_{j}(t)}}{\sum_{j}{V_{j}(t)}}$ edges are removed.

The expected decrease of check nodes with degree $i$ is written as,

\begin{multline}
\mathbb{E}\left[C_{i}(t+1)-C_{i}(t)\lvert V(t),E(t) \right] \approx\\
\frac{i \left[C_{i+1}(t)-C_{i}(t)\right]}{\sum_{j}{j
V_{j}(t)}}\frac{\sum_{j}{j (j-1) V_{j}(t)}}{\sum_{j} j V_{j}(t)},
i \ge 2
\end{multline}

The examined node is removed from the set of check nodes of degree one, and the expected decrease in degree is given as
\begin{multline}
\mathbb{E}\left[C_{1}(t+1)-C_{1}(t)\lvert V(t),E(t) \right] \approx\\
-1+\frac{\left[C_{2}(t)-C_{1}(t)\right]}{\sum_{j}{j
V_{j}(t)}}\frac{\sum_{j}{j (j-1) V_{j}(t)}}{\sum_{j} j V_{j}(t)}
\end{multline}

Since we assume that the process stays close to its expected behavior, as the decoding algorithm corresponds to a stochastic process, we can drop the expectation and write $v_{i}(\tau) \approx \frac{V_{i}(t)}{n V'(1)}$ and $c_{i}(\tau) \approx \frac{C_{i}(t)}{n V'(1)}$, where $\tau \triangleq \frac{t}{n V'(1)}$ denotes the normalized time. $V(t)$ is the degree distribution function of the variable nodes and $V(1)^{'}$ is the average degree of $V(t)$. Formally, $c_i$ (or $v_i$) describes the fraction of check nodes (or variable nodes) connected to $i$ edges. Then, we make the following approximations

\begin{equation}
\frac{V_{i}(t+\delta t)-V_{i}(t)}{\delta t} \approx \frac{\partial
v_{i}(\tau)}{\partial t}
\end{equation}
and
\begin{equation}
\frac{C_{i}(t+\delta t)-C_{i}(t)}{\delta t} \approx \frac{\partial
c_{i}(\tau)}{\partial t}
\end{equation}

\begin{figure*}[t]
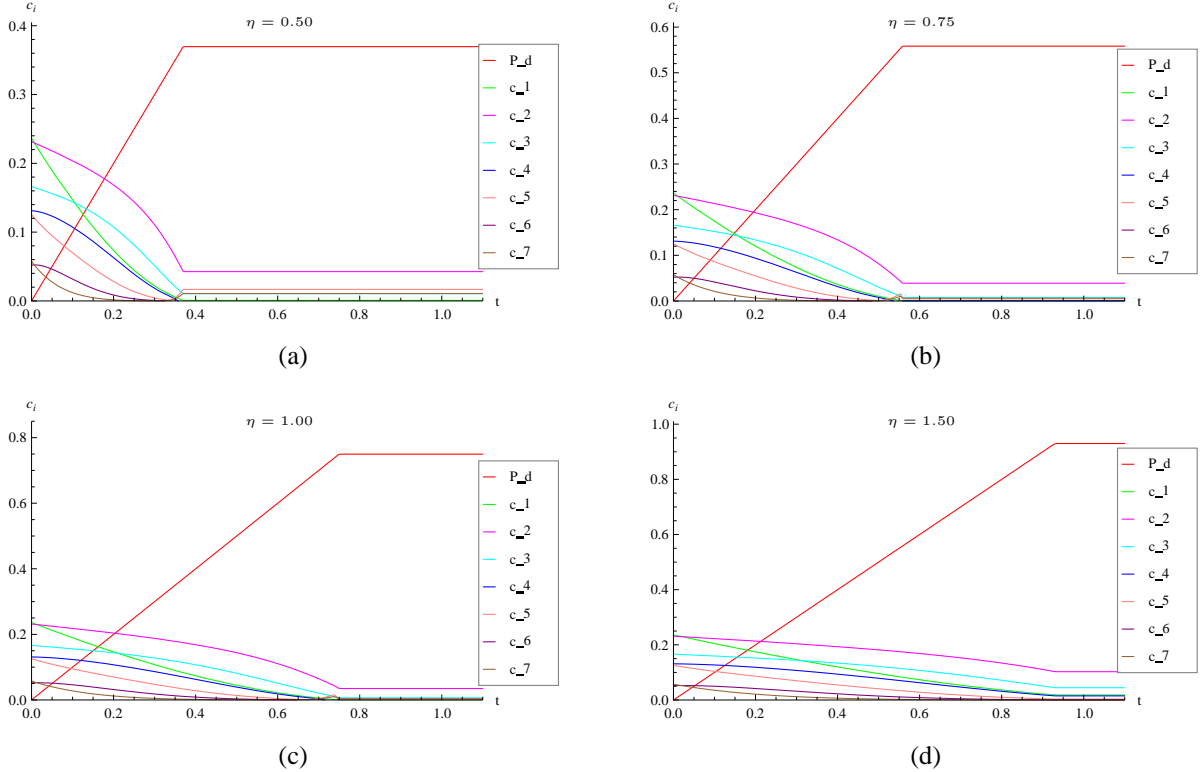

\begin{center}
\begin{tabular}{cc}
~\input{subfigure_result_R_0p5.pgf}
~ & ~\input{subfigure_result_R_0p75.pgf}~\\
~(a)~ & ~(b) ~\\
~\input{subfigure_result_R_1p0.pgf}
~ & ~\input{subfigure_result_R_1p5.pgf}~\\
~(c)~ & ~(d) ~\\
\end{tabular}
\end{center}
\caption{Temporal evolution of the fraction $c_i(t)$ of check nodes of degree $i$  for various reception rates $\eta$: (a) $\eta=0.5$, (b) $\eta=0.75$, (c) $\eta=1.0$, and (d) $\eta=1.5$.
}
 \label{fig:growth_analysis_deg}
\end{figure*}

We can now express the differential equations involving variable and check nodes degree distributions, which govern the behavior of the decoder:

\begin{eqnarray}
\frac{\partial v_{i}(t)}{\partial t} &=& -i \frac{v_{i}(t)}{\sum_{j=1}^{d_{v}^{\max}}{{j v_{j}(t)}}}, \label{eqn:differential1} \\
\frac{\partial c_{i}(t)}{\partial t} &=& \left[c_{i+1}(t)-c_{i}(t)\right]\frac{i \sum_{j}^{d_{c}^{\max}}{{j (j-1)}c_{j}(t)}}{\sum_{j=1}^{d_{c}^{\max}}{{j c_{j}(t)}}},  \label{eqn:differential2} \\
\frac{\partial c_{1}(t)}{\partial t} &=& -1+\left[c_{2}(t)-c_{1}(t)\right]\frac{\sum_{j}^{d_{c}^{\max}}{{j (j-1)}c_{j}(t)}}{\sum_{j=1}^{d_{c}^{\max}}{{j c_{j}(t)}}}\;\;\; \;\label{eqn:differential3} \\
\nonumber
\end{eqnarray}

The parameters $d_c^{\max}$ and $d_v^{\max}$ are the maximum degrees of the check and variable nodes respectively; we typically set them to a very large value (theoretically infinity). Eq. (\ref{eqn:differential1}) shows the expected decrease in the number of variable nodes of degree $i$. Eq. (\ref{eqn:differential2}), which is valid $\forall i, i\ge2$, and Eq. (\ref{eqn:differential3}) describe the residual degree changes from $i+1$ to $i$ when an edge connected with a check node of degree $i+1$ is removed. The first term ($-1$) in Eq. (\ref{eqn:differential3}) is due to the removal of an edge from the check node of degree one.

The solution of these differential equations with appropriate initial conditions describe the evolution of the degree distribution in the check and variable nodes. For Growth codes, we have that the initial check nodes distribution of the check nodes is $\Omega(d)$ in Eq. (\ref{eqn:growth}). Then, in the limit of large number of source packets $k$, the variable nodes follow a Poisson distribution \cite{RandomGraphs}, since the source symbols are selected uniformly at random during the encoding process.

A closed form solution to these equations is non-trivial to determine. However, we can solve them numerically. For verifying our analysis we examine the degree evolution of the nodes during decoding with a Decoder-D algorithm. We choose the receiving symbol rate  $\eta$ (normalized to the source rate) in the range $[0, 1.5]$ (\textit{i.e.}, for $\eta < 1$ we have intermediate performance). The evolution of various degree edges at different $\eta$ values is shown in Fig. \ref{fig:growth_analysis_deg} where the number $k$ of source packets is set to 1000. The fraction of variable nodes of degree $0$ corresponds to the recovery probability $P_d$, i.e., $P_d=v_0$. The evolution is in accordance with the results presented in \cite{Growthcodes}. From Fig. \ref{fig:growth_analysis_deg}, it becomes obvious that the number of nodes with degree zero increases with the decoding time, but it reaches a stationary point that determines the performance limit of the codes. We can also observe that the performance improves with the symbol rate. Furthermore, our model based on Wormald method closely matches the simulation results in \cite{Growthcodes}, which are derived specifically for moderate codeblocks such as $k=1000$. Our method is however more generic and applies to different settings. For example, the proposed analysis method is appropriate for large codeblocks and medium size codeblocks, \textit{i.e.}, $k < 1000$ while other analysis works \cite{Sanghavi07} can be used exclusively for large codeblocks $k > 10000$. Interestingly, the presented analysis method is valid for any $\eta$ value. In contrary to our method, the work in \cite{Sanghavi07} can be applied only in the intermediate range.

\subsection{Recovery probability }
\label{sec:symdecprob}

The above analysis can be used to characterize the decoding procedure and derive the symbol decoding probability for the Decoder-D. We thus define the symbol recovery probability as $\mathbb{P}_d=v_0$ ($v_0$ stands for the probability of recovering a symbol, its degree is zero). The probability $\mathbb{P}_d$ is depicted in Fig. \ref{fig:growth:symbol_recovery:probability4} for various code rates $\eta$. The dotted line represents the symbol recovery performance for Growth codes derived by the Wormald method, the dashed-doted line represents the results of the experimental evaluation for $k=1000$ and the dashed line describes the recovery performance of an ``ideal'' code that is able to recover a new source symbol with every received Growth encoded symbol. The ``ideal'' code is presented for the sake of the completeness, since such code does not exist in practice. From the Fig. \ref{fig:growth:symbol_recovery:probability4}, we can infer that, for large receiving symbol rates ($\eta>1.5$), Growth codes can decode the full message as $\mathbb{P}_d$ reaches one.

To the best of our knowledge, no closed formula exists for the symbol decoding probability in the Decoder-D algorithm, which is valid for any arbitrary check nodes degree distribution; it is quite difficult to derive such an analytical model from the differential equations (\ref{eqn:differential1})-(\ref{eqn:differential3}). By observing Fig. \ref {fig:growth:symbol_recovery:probability4}, we however note that an exponential model is a good fit. Thus, the symbol decoding probability can be approximated as 
\begin{equation}
\mathbb{P}_{d}(\eta)  \approx 1-\lambda e^{-\mu \eta^{2}}
\label{eqn:Pd}
\end{equation}
where $\eta=r/k$ is the coderate and $r$ denotes the number of received symbols. Similarly, we define the probability of symbol decoding failure as
\begin{equation}
\mathbb{P}_{l}(\eta)  = \lambda e^{-\mu \eta^{2}}.
\label{eqn:Pl}
\end{equation}

The values of $\mu$ and $\lambda$ depend on the underlying degree distribution function and not on the code length. This is expected as our model based on Wormald method also depends only on the degree distribution of the employed codes. We found through experimentation that for Growth codes $\lambda =0.926854$ and $\mu=1.39361$. The evaluation of the model of Eq. (\ref{eqn:Pd}) is depicted with the continuous line in Fig. \ref{fig:growth:symbol_recovery:probability4} for $k=1000$. Since our model is a function of the $\eta$ value and not of the code length $k$, we can use it to characterize the performance of Growth codes of various codeblocks. We see from Fig. \ref{fig:growth:symbol_recovery:probability4} that our model agrees with the experimental evaluation. 

\begin{figure}[t]
\begin{center}
\includegraphics[width=0.5\textwidth]{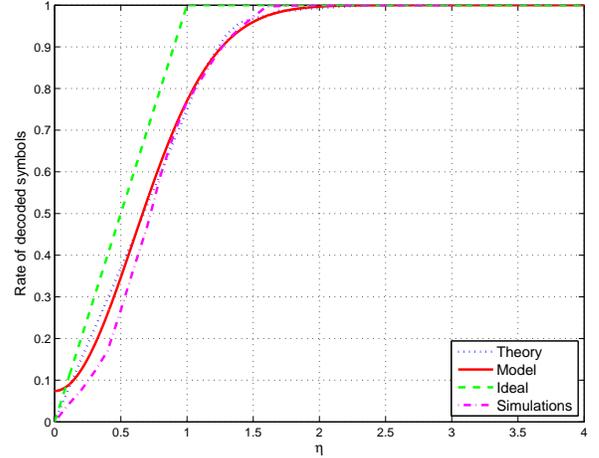}
\end{center}
\vspace{-0.5cm}
\caption{Symbol recovery probability of Growth codes. The continuous curve is $\mathbb{P}_{d}(\eta)  = 1-\lambda e^{-\mu \eta^{2}}$ ($\eta$ denotes the reception rate), the dotted line corresponds to the growth code performance derived by Wormald method, the dashed-doted line represents the results of the experimental evaluation, and the dashed line the ideal intermediate performance (it corresponds to the performance of an ``ideal'' code that is able to recover with every received Growth symbol a new source symbol). The codeblock size $k$ is equal to 1000.}
\label{fig:growth:symbol_recovery:probability4}
\end{figure}

\section{Error-resilient Video Streaming}
\label{sec:distanalysis}

The analysis of the Growth Codes performance is important for the optimization of error resilient video communications where the good intermediate performance of the Growth codes can be very beneficial. Video is indeed a signal that may tolerate some losses and still offer a decent rendering quality. In particular, the above permits to select appropriate source and channel rates for optimized communications on rate-constrained lossy channels. The efficiency of video transmission systems is measured with respect to the video quality at the clients. This quality typically depends on the encoded video at sources and on the effect of losses during transmission that are respectively driven by the source and channel rates. In the following, we quantify the end-to-end performance of  an illustrative video application using Growth codes in order to highlight the benefits of good intermediate performance. Note that a similar analysis can be done for LT codes, which also fit with the representation used in the Wormald analysis above.

\subsection{Illustrative video model}
\label{sec:videomodel}

Video compression is first applied to an image sequence in order to reduce the transmission rate. However, compression renders video data  sensitive to packet losses. Therefore, the compressed video data should be protected to avoid rapid quality degradation in case of packet losses. Typically, this is achieved by forward error correction coding (FEC) schemes.

The overall distortion for communication scenarios where channel codes are used for FEC depends on the efficiency of the employed source compression scheme and the channel codes. Whereas the source coding distortion $D_{s}$ decreases with increasing source rate $R$, the channel coding distortion $D_{c}$ increases with lesser redundancy. Several distortion models are proposed in the literature \cite{paper:frossard:2001}, \cite{paper:Sakazawa}. Considering hierarchically structured compression schemes such as MPEG, an analytical expression has been found in \cite{paper:frossard:2001} for the end-to-end distortion $D$, where the distortion depends on the MPEG coding parameters and the average recovery probability of the deployed FEC codes.

The source distortion can be written as
\begin{equation}
D_{s}(R)=\alpha R^{-\beta}
\label{eqn:Ds}
\end{equation}
where $R$ is the source bit rate, $\alpha>0$ and $\beta>0$ are parameters that depend on the source encoder and on the video content.

The distortion, in case of loss, takes into account error propagation, error patterns and error concealment. Therefore, it reads
\begin{equation}
D_{L}(R)= b \left(1+ \frac{R}{2 a N_{s} L_{p} }\right)  \mathbb{P}_{l}(\eta)
\label{eqn:Dc}
\end{equation}
where $N_{s}$ is the average number of video slices\footnote{Each slice consists of several macroblocks, which are the basic encoding unit of H.264/AVC and other block-based video compression schemes.} per second and $L_{p}$ is the packet size (in bytes) of the transmitted video data. The parameter $a$ indicates the type of video loss/distortion pattern that depends on the packetization scheme; $b$ is a constant related to the spatio-temporal complexity of the sequence and error concealment scheme. Finally, $\mathbb{P}_{l}(\eta)$ is the average residual packet loss probability after FEC decoding and $\eta=r/k$ is a redundancy factor with $r$ being the number of channel symbols and $k$ the number of source symbols. We propose to use this end-to-end distortion model in our illustrative example of optimized Growth Codes video transmission scheme, where we replace the packet loss probability $\mathbb{P}_{l}(\eta)$ by the one obtained in the performance analysis of Section \ref{sec:wormald}. 

\subsection{Optimal joint source and channel coding}
\label{sec:joint}

\begin{figure}[t]
\begin{center}
\includegraphics[width=0.5\textwidth]{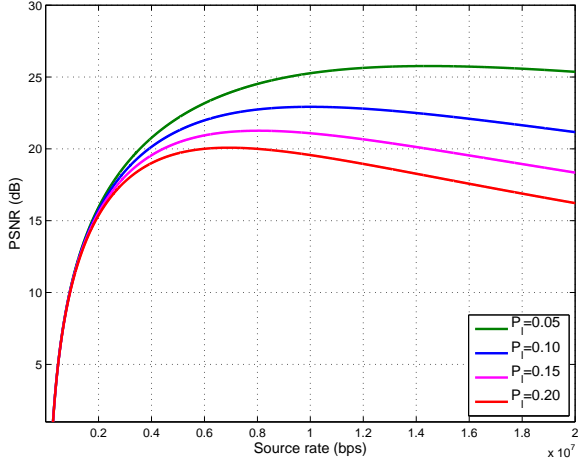}
\end{center}
\caption{Illustration of the PSNR (in dB) for the Foreman CIF sequence encoded at different rates ($R_s$) for various decoding failure probabilities $\mathbb{P}_l(n)$ using the model of Eq. (\ref{eq:optimal_D}).}
\label{fig:fig_video_D_convex1}
\end{figure}

Since the video quality depends on the source and channel distortions, it depends on the source bit rate $R$. Thus, the optimal value of $R$ that minimizes the end-to-end distortion function $D(R)$ can be determined by

\begin{eqnarray}
R^{\star} =\arg \underset{R}{\min} \quad {D(R)}\\
s.t. \; (1+\epsilon)R \le B \nonumber
\label{eqn:Rstar}
\end{eqnarray}

\noindent where
\begin{eqnarray}
D(R) &=& \left(1-\mathbb{P}_{l}(\eta)\right) D_{s}(R)+\mathbb{P}_{l}(\eta)  D_{c}(R) \\
    &=& D_{s}(R)+\mathbb{P}_{l}(\eta) \left( D_{c}(R)- D_{s}(R)\right) \nonumber\\
    &=& D_{s}(R) + D_{L}(R) \nonumber
\end{eqnarray}
and $D_c(R)$ denotes the concealment distortion (\textit{i.e.}, the average distortion between original macroblocks and the concealed marcoblocks at the receiver in case of loss) \cite{paper:frossard:2001} and $D_L(R)=\mathbb{P}_{l}(\eta) \left( D_{c}(R)- D_{s}(R)\right)$ the average distortion between a losslessly and lossily transmitted versions, as defined in \cite{paper:frossard:2001}. Trivially, for lossless transmission, $D_c$(R) is equal to zero. Here, the parameter $\mathbb{P}_l(\eta)$ is a constant where the affordable loss rate is determined by the video application constraints. The parameter $B$ in Eq. (\ref{eqn:Rstar}) stands for the channel capacity, whose constraint takes into account \textit{i.e.}, the amount of added channel redundancy (i.e., $\epsilon$) for achieving a given symbol decoding failure rate  $\mathbb{P}_l$. The value of $\epsilon$ can be computed from Eq. (\ref{eqn:Pl}) by setting $\eta=1+\epsilon$. We further assume that there is no estimate for $\pi$, the actual channel packet loss rate. Obviously, when an estimate for $\pi$ is available, other Rateless codes such as LT codes would be a better solution, as they characterized by smaller overhead value $\epsilon$ than Growth codes. Growth codes are however interesting in a setup where the actual channel loss rate is unknown, and where the application can benefit from improved intermediate performance. 

Now, taking into account Eqs. (\ref{eqn:Ds}) and (\ref{eqn:Dc}), we can write Eq. (\ref{eqn:Rstar}) as
\begin{equation}
D(R) = \alpha R^{-\beta} + b \left(1+ \frac{R}{2 a N_{s} L_{p} }\right) \mathbb{P}_{l}(\eta)
\label{eqn:D}
\end{equation}

Then, the source and channel rate optimization problem of Eq. (\ref{eqn:Rstar}) is rewritten as
\begin{equation}
R^{\star} =\arg \underset{R}{\min} \quad {{ \left[\alpha R^{-\beta} + b \left(1+ \frac{R}{2 a N_{s} L_{p} }\right)  \mathbb{P}_{l}(\eta) \right]} }.
\end{equation}

Before determining the optimal source coding rate $R^{\star}$, we first prove that $D(R)$ is convex.

\begin{lemma}
The end-to-end distortion function $D(R)$ is convex.
\label{l:lemma}
\end{lemma}

\begin{proof}
In order to prove the convexity of $D(R)$ is convex, we compute the first derivative of $D(R)$ with respect to the rate $R$.

\begin{eqnarray*}
\frac{\partial}{\partial R} D(R) &=& \frac{\partial}{\partial R} { \left[\alpha R^{-\beta} + b \left(1+ \frac{R}{2 a N_{s} L_{p} }\right)  \mathbb{P}_{l}(\eta) \right]} \\
&=& \frac{b \mathbb{P}_{l}(\eta) }{2 a N_{s} L_{p} }-\alpha \beta R^{-(1+\beta)}
\end{eqnarray*}

Therefore, the stationary point $R^{\star}$ that satisfies $\frac{\partial}{\partial R} D(R) =0$ is,
\begin{equation}
R^{\star}= \left(\frac{2 a \alpha \beta N_{s} L_{p}}{b \mathbb{P}_{l}(\eta)}\right)^{\frac{1}{1+\beta}}
\label{eq:optR}
\end{equation}

We now have to prove that the second derivative $\frac{\partial^{2}}{\partial R^{2}} D(R) $ at $R^{\star}= \left(\frac{2 a \alpha \beta N_{s} L_{p}}{b \mathbb{P}_{l}(\eta)}\right)^{\frac{1}{1+\beta}}$  is non-negative.

Thus,
\begin{eqnarray*}
\frac{\partial^{2}}{\partial R^{2}} D(R) &=& \alpha \beta \left(1+\beta\right) (R^{\star})^{-(2+\beta)} \\
&=& \alpha \beta \left(1+\beta\right) \left(\frac{b \mathbb{P}_{l}(\eta)}{2 a \alpha \beta N_{s} L_{p}}\right)^{\frac{2+\beta}{1+\beta}},
\label{eq:rstar}
\end{eqnarray*}
which is clearly positive since $\alpha,\beta>0$. Hence, $D(R)$ is convex and a unique minimum exists for a fixed $\mathbb{P}_{l}(\eta)$.
\end{proof}

Therefore, we can conclude the following. When $(1+\epsilon)R^{\star} \ge B$, the optimal rate $R^{\star}$ is computed by Eq. (\ref{eq:optR}); otherwise, $R^{\star}=\frac{B}{1+\epsilon}$. This is due to the fact that the plane $(1+\epsilon)R - B$ intersects $D(R)$ at $R=\frac{B}{1+\epsilon}$.

\begin{figure}[t]
\begin{center}
\includegraphics[width=0.5\textwidth]{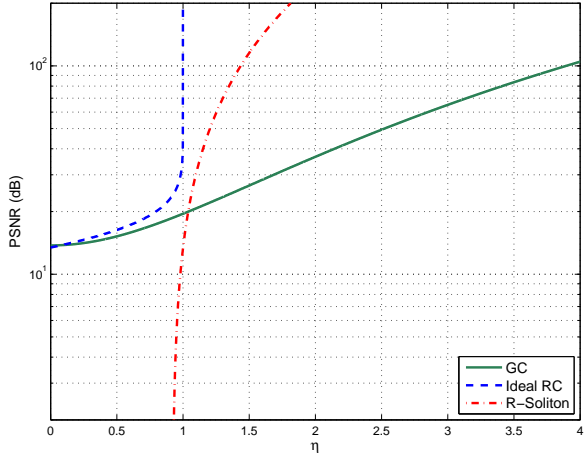}
\vspace{-0.5cm}
\end{center}
\caption{Comparison of the proposed optimized Growth codes (GC) with ideal Rateless codes (ideal RC) and LT codes ``R-Soliton'' for transmission of Foreman CIF sequence with respect to various $\eta$ values. The distortion is given in terms of PSNR.}
\label{fig:fig_growth_video_curves1}
\end{figure}

Hence, we have that the minimum distortion ${D}^{\star}$ is given by,
\begin{equation}
D^{\star}=\alpha (R^{\star})^{\beta}+b\left(1+\frac{R^{\star}}{2 a N_{s} L_{p}}\right) \lambda e^{-\mu \eta^{2}}
\label{eq:optimal_D}
\end{equation}
where the optimal rate is equal to
\begin{equation}
R^{\star}=\left( \frac{2 a \alpha \beta N_{s} L_{p}}{b\lambda e^{-\mu \eta^{2}}}\right)^{\frac{1}{1+\beta}}.
\label{eq:optimal_R}
\end{equation}

The distortion for video transmission under different packet recovery probabilities, when Growth codes are used for error protection and the model parameters of \cite{paper:frossard:2001} are adopted is shown in Fig. \ref{fig:fig_video_D_convex1}. We investigate the distortion in both intermediate region $0\le \eta <1$ and full recovery region $\eta \ge 1$. From Fig. \ref{fig:fig_video_D_convex1}, we can see that $D(R)$ is strictly convex. 

\subsection{Evaluation of rateless video distortion}

We are interested in the average distortion performance when a certain number of coded packets reach the receiver. No synchronization between the servers and receivers is assumed, which means that some packets may be received multiple times. 

Distortion comparisons in terms of PSNR are illustrated in Fig. \ref{fig:fig_growth_video_curves1} when the $D(R)$ convex hull does not intersects with the plane $(1+\epsilon)R - B$, \textit{i.e.}, there is enough bandwidth to accommodate the transmission of the video encoded at rate computed by Eq. (\ref{eq:optR}). The source coding parameters are chosen based on the model in \cite{paper:frossard:2001}. The proposed rate optimal Growth codes are denoted as GC in Fig. \ref{fig:fig_growth_video_curves1}. The optimal source rate value is calculated by Eq. (\ref{eq:optR}). The proposed scheme is compared with an ideal Rateless code, which can recover as many source symbols as the received coded symbols, \textit{i.e.}, it recovers $\eta N$ source symbols when $\eta N$ coded packets are received in the intermediate region (\textit{i.e.,} $0 \le \eta < 1$) and fully recovers the source when $\eta \ge 1$. Clearly, the performance of the ideal code serves as an upper limit of the performance of any channel code. From Fig. \ref{fig:fig_growth_video_curves1}, we note that when the proposed rate optimal Growth codes are used, the overall distortion follows a smooth waterfall curve. In the intermediate region, the rate optimal Growth codes offer comparable distortion performance to ideal Rateless codes. We also compare with LT codes (denoted as ``R-Solition''). The optimal source rate is computed by Eq. (\ref{eq:optR}) using our model (Eq. (\ref{eqn:Pl})) with parameter values $\lambda=5\cdot 10^8$ and $\mu=-20$ for calculating the decoding failure probability. The large distortion gap between Growth codes and LT codes (denoted as ``R-Soliton'') in intermediate region is also obvious from Fig. \ref{fig:fig_growth_video_curves1}. This is expected as the Robust soliton distribution has a small percentage of degree one symbols and it is designed for full recovery. In the full recovery area, the performance of LT codes improves rapidly with $\eta$ value and approaches the performance of ideal Rateless codes. Hence, we can conclude the full advantage of Growth codes is expressed with a source coder such as a video coder whose utility increases with the decoded data. 

\begin{figure}[t]
\begin{center}
\includegraphics[width=0.5\textwidth]{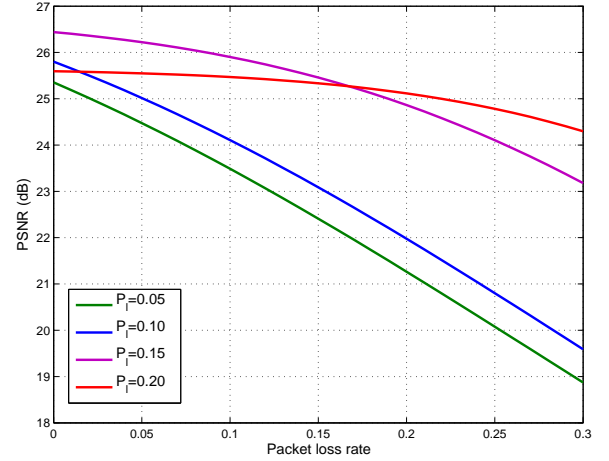}
\end{center}
\caption{Evaluation of the proposed scheme for transmission of the Foreman CIF sequence over different values of loss rate and different acceptable packet error rates $\mathbb{P}_l$ when channel capacity $B$ is $15 \;Mbps$.}
\label{fig:fig_mismatch}
\end{figure}

For the sake of completeness, we provide comparisons for channel mismatch conditions, \textit{i.e.}, when the rate of GC is optimized for a set of acceptable decoding failure rates $(5, 10, 15, 20)$\% ( these values are considered when the optimal source rates are calculated), but in practice the channel conditions are different than those considered during optimization. Specifically, the source rate is optimized assuming that the channel is free of packet losses, however in practice the actual loss rate varies in the range of $[0,30]$\%. The results are presented in Fig. \ref{fig:fig_mismatch} for transmission of the Foreman CIF sequence and channel capacity $B=15$ $Mbps$. From the results, we can observe that the acceptable decoding failure rate plays important role on the robustness of the method. When the acceptable loss rate is 15\% or 20\%, the optimal source rate is significantly lower than the capacity $B$. Hence, our scheme is characterized by higher $\eta$ values and our method is more resilient to variations of the loss rate, since the packet decoding failure rate is not affected significantly. This is verified by Fig. \ref{fig:fig_mismatch2} where we can see that the packet failure probability changes smoothly. However, when the acceptable loss rate is low, the optimal source rate is rather high and the scheme is optimized for lower $\eta$ values. In such case the symbol decoding rate is affected significantly by the increased loss rate. This is due to the fact the actual $\eta$ value is smaller than the theoretical optimal because of the channel losses. Therefore, the video transmission becomes more fragile to error changes. It is worth noting that when the acceptable failure rate is 5\%,  the packet decoding failure rate changes from 5\% to 22\%. 

\begin{figure}[t]
\begin{center}
\includegraphics[width=0.5\textwidth]{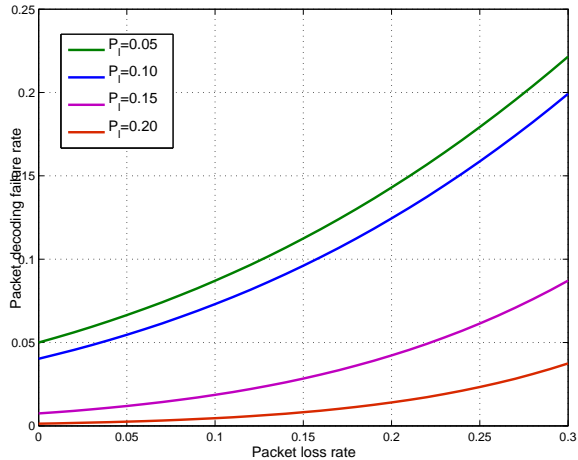}
\end{center}
\caption{Decoding packet failure probability for different acceptable packet error rates $\mathbb{P}_l$ when channel capacity $B$ is $15 \;Mbps$. Various loss rates are examined.}
\label{fig:fig_mismatch2}
\end{figure}

\section{Conclusions}
\label{sec:conclusions}

In this paper, we have analyzed the performance of Growth codes using the Wormald method and derived generic bounds that fully characterize the decoding process. We find that an exponential model can nicely describe the decoding performance. In contrary to other schemes in the literature, the proposed analysis framework is appropriate for medium sized codeblocks and not only to large sized codeblocks. It is general and appropriate for analyzing any other Rateless code that does not include a precoding step. We proposed an illustrative video streaming application to demonstrate the benefits of the good intermediate performance of Growth codes. We have casted a joint source and channel rate allocation problem whose convex objective function permits to select of the good intermediate performance of Growth codes. Our illustrative experiments show that Growth codes offer good intermediate performance that are useful for video applications whose quality monotonically grows with the number of packets. Finally, Growth codes are pretty robust to inaccurate estimations of the channel status, which is very interesting in the design of practical streaming applications.

\end{document}